\documentclass[twocolumn,aps,pra,showkeys,showpacs,groupedaddress,nofootinbib,longbibliography]{revtex4-1}

\input{preamble}

\begin{document}

\title{Classical theories with entanglement}
\date{\today}
\author{Giacomo Mauro D'Ariano}
\email{dariano@unipv.it}
\author{Marco Erba}
\email{marco.erba@unipv.it}
\author{Paolo Perinotti}
\email{paolo.perinotti@unipv.it}
\affiliation{Universit\`a degli Studi di Pavia, Dipartimento di Fisica, QUIT Group, and INFN Gruppo IV, Sezione di Pavia, via Bassi 6, 27100 Pavia, Italy}
\begin{abstract}
We investigate operational probabilistic theories where the pure states of every system are the vertices of a simplex. A special case of such theories is that of \emph{classical theories}, i.e.~simplicial theories whose pure states are jointly perfectly discriminable. The usual Classical Theory satisfies also local discriminability. However, simplicial theories---including the classical ones---can violate local discriminability, thus admitting of entangled states. First, we prove sufficient conditions for the presence of entangled states in arbitrary probabilistic theories. Then, we prove that simplicial theories are necessarily causal, and this represents a no-go theorem for conceiving non-causal classical theories. We then provide necessary and sufficient conditions for simplicial theories to exhibit entanglement, and classify their system-composition rules. We conclude proving that, in simplicial theories, an operational formulation of the superposition principle cannot be satisfied, and that---under the hypothesis of $n$-local discriminability---no mixed state admits of a purification. Our results hold also in the general case where the sets of states fail to be convex.
\end{abstract}
\maketitle

\section{Introduction}
Entanglement is the quantum feature marking the starkest departure of Quantum Theory from Classical Theory (in the following, these will be referred to as QT and CT, respectively). The phenomenon is commonly popularised as the so-called \emph{quantum nonlocality}, although  the two concepts are not coincident. 
Indeed, it is known that the mere existence of entangled states is not sufficient for nonlocality
~\cite
{PhysRevA.40.4277,PhysRevA.65.042302}. On the other hand, CT does not allow for any kind of entanglement or nonlocality.

States in CT have a very simple geometrical structure, namely the set of states for every system is a \emph{simplex}. In the present work, we argue that the absence of entanglement in CT is due not only to the simplicial structure of the set of states of the single systems, but also to the composition rule of the systems, which satisfies local discriminability~\cite{bookDCP2017}. Indeed, as we will show, {\em simplicial theories} can exhibit entanglement,  including general {\em classical theories}, namely those simplicial theories where the pure  states are jointly perfectly discriminable~\cite{barrett2007information,PhysRevLett.99.240501,pfister2013information}.  Besides classical theories, the definition of simplicial theory encompasses
even more general cases---e.g.~noisy versions of classical theories, where pure states cannot be reliably distinguished due to a limited set of measurements \cite{janotta2014generalized}. The characteristic trait of a simplicial theory is that every state has a unique convex decomposition into pure states. The results of the present work hold for all simplicial theories, encompassing even cases where the sets of states are not convex. In contrast to the common understanding of the quantum/classical divide, as a consequence of our results we show that a classical theory generally admits of entangled states.

The paper is organised as follows. In Section~\ref{sec:OPTs} we provide a review of the framework of Operational Probabilistic Theories (OPTs). In Section~\ref{sec:entangl_OPTs} we point out the relevance of the notion of parallel composition in a probabilistic theory, focusing on the consequences regarding the presence of entanglement in arbitrary OPTs. In particular, we prove sufficient conditions for the presence of entangled states in a probabilistic theory. In Section~\ref{sec:entangl_simpl}, we specialise to simplicial theories, discussing various features of this class of OPTs. First, as a consequence of the geometric structure of the sets of states, we show that a simplicial theory is necessarily causal, highlighting that causality is intrinsic to classical theories---and, in particular, to CT. Then, we provide necessary and sufficient conditions for the presence of entangled states in a simplicial theory, along with a classification of simplicial theories in terms of their composition rules for systems. We conclude showing that, despite the possibility of exhibiting entanglement, no simplicial theory can admit of superposition (in an operational sense), or 
purifications of (any) mixed states (under the hypothesis of $n$-local discriminability). Finally, in Section~\ref{sec:conclusions} we discuss some information-theoretic features shared by simplicial theories with entanglement, making a comparison with the existing literature on probabilistic theories, and we draw our conclusions.

\section{Review of the framework of Operational Probabilistic Theories}\label{sec:OPTs}
The primitive notions of an operational theory are those of \emph{systems}, \emph{tests}, and \emph{events}. A system $\sys{S}$ represents the physical entity which is probed in the laboratory, such as a radiation field, a molecule, an elementary particle. A test $\T{E}$ is characterised by a collection of events $\lbrace \T{E}_x \rbrace_{x\in\rX}$, and represents the single occurrence of a physical process, such as the use of a physical device or a measuring apparatus. The \emph{outcome space} $\rX$ associated with a test $\T{E}$ collects all the possible outcomes that can occur within the test. Each test $\T{E}$ is characterised by an input system $\sys{A}$ and an output system $\sys{B}$. For instance, think of an electron--proton scattering: both the input and the output systems are an electron--proton pair, the test contains only one event corresponding to the two-particle interaction, and finally the outcome space is a singleton. On the other hand, a fundamental feature for a physical theory is the power to make predictions. Accordingly, the operational structure needs to be endowed with a suitable set of rules for computing the probability distributions associated to each experiment.

The above structure can be formalised defining the notion of an Operational Probabilistic Theory (OPT).\footnote{For a thorough presentation of the OPT framework, we refer the interested reader to Refs.~\cite{PhysRevA.81.062348,PhysRevA.84.012311,bookDCP2017,chiribella2016quantum}} In the present section, we recall some relevant concepts in the scope of OPTs. Given an OPT $\Theta$, $\Sys{\Theta}$ will denote the set of the systems of the theory. We will denote by Roman letters $\sys{A,B,\ldots}\in\Sys{\Theta}$
the systems, and by $\T{E}=\lbrace \T{E}_x \rbrace_{x\in\rX}$ the tests, where each $\T{E}_x$ for $x\in\rX$ is a possible event that can occur in the test. Given a test $\T{E}$, each possible event $\T{E}_x$ is a map
from an input system $\rA$ to an output system $\rB$, and corresponds to the outcome $x$. Whenever the output of a test $\T{E}_1=\lbrace \T{E}_x \rbrace_{x\in\rX}$ coincides with the input of another test $\T{E}_2=\lbrace \T{E}_y \rbrace_{y\in\rY}$, the \emph{sequential composition} $\T{E}_2\circ\T{E}_1\coloneqq\T{E}_2\T{E}_1\equiv\lbrace \T{E}_y\T{E}_x \rbrace_{(y,x)\in\rY\times\rX}$ of the two tests can be defined, being an allowed test for the theory. Thus, physical systems define the connection rules for tests.
In an OPT, two important classes of tests are those of \emph{preparations} and \emph{observations}, i.e.~tests with no input or no output system, respectively. Accordingly, preparations $\rho$ and observations $a$, where  $\sys{A}$ is, respectively, the output or the input system, are conveniently denoted using the Dirac-like notation $\rket{\rho}_{\sys{A}}$ and $\rbra{a}_{\sys{A}}$. Given some arbitrary preparation event $\rket{\rho_i}_{\sys{A}}$, event $\T{E}_x$ from system $\sys{A}$ to system $\sys{B}$, and observation event $\rbra{a_k}_{\sys{B}}$, the purpose of an OPT is to compute joint probabilities of the form:
\begin{align}\label{eq:probability_OPT}
p\left(i,x,k\left.|\right. \rho,\T{E},a \right) \coloneqq \rbra{a_k}_{\sys{B}}\T{E}_x\rket{\rho_i}_{\sys{A}}.
\end{align}
Two events from a system $\sys{X}$ to a system $\sys{Y}$ are equivalent if all their joint probabilities of the form~\eqref{eq:probability_OPT}, given the same set of other events appearing in Eq.~\eqref{eq:probability_OPT}, are equal. We will call equivalent classes of preparation and observation events \emph{states} and \emph{effects}, respectively. Besides, equivalent classes of arbitrary events from $\sys{A}$ to $\sys{B}$ will be called \emph{transformations}. Notice that states and effects will be considered special cases of transformations. This can be done simply by introducing the notion of \emph{trivial system} $\sys{I}$, defined as their, respectively, input and output system. For every $\sys{A},\sys{B}\in\Sys{\Theta}$, $\St{A}$, $\Transf{A}{B}$, and  $\Eff{B}$ will denote, respectively, the sets of states of $\sys{A}$, of transformations from $\sys{A}$ to $\sys{B}$, and of effects of $\sys{B}$. Clearly, every event $\T{E}_x\in\Transf{A}{B}$ is a map from $\St{A}$ to $\St{B}$, or, dually, from $\Eff{B}$ to $\Eff{A}$. According to the above definitions, effects are separating for states, i.e.~given two states $\rket{\rho_1}_{\sys{A}}\neq\rket{\rho_2}_{\sys{A}}$, there exists $\rbra{a}_{\sys{A}}\in\Eff{A}$ such that $\rbraket{a}{\rho_1}_\sys{A}\neq\rbraket{a}{\rho_2}_\sys{A}$. Similarly, an effect is the equivalence class of those observation events that give the same probabilities for every state, and thus states are separating for effects.

Within an operational perspective, an agent is allowed to perform a test---say with outcome space $X$---disregarding the single outcomes within a subset $Y\subseteq X$, and then merging events in $Y$ into a single event: this possibility is captured by the notion of \emph{coarse-graining}. According to probability theory, the probability of the coarse-grained event $Y$ amounts to the sum of the probabilities of all the outcomes in the subset $Y$. Then, for each test $\lbrace \T{E}_x \rbrace_{x\in X}\subseteq \Transf{A}{B}$ and every subset $Y\subseteq X$, the coarse-grained event is symbolically given by $\sum_{y\in Y}\T{E}_y$, where sequential composition distributes over sums. The converse procedure of a coarse-graining is called a \emph{refinement}. An event with trivial refinement, i.e.~which cannot be further refined modulo a rescaling by a probability, is called \emph{atomic}. Being the framework probabilistic, one may also want to consider convex combinations of states, transformations, and effects, corresponding to a randomization, i.e.~a statistical mixture of events.

In general, states, transformations, and effects can be thought of as embedded in convex spaces. However, in general the set of states of a system might not be convex. Then, one can conveniently remind two important notions from conic and convex analysis, i.e.~those of atomic and extremal points. Let $\St{A}$ be the set of states of a system $\sys A$. The null state $\rket{\varepsilon}_\sys{A}$, which occurs with null probability in every context, is always included in $\St{A}$. Given $\rket{x_1},\rket{x_2}\in\St{A}$ and $p\in(0,1)$,
a state $\rket{x}\in\St{A}$ is called \emph{atomic} if the condition $\rket{x}=\rket{x_1}+\rket{x_2}$ implies $\rket{x_1}\propto \rket{x_2}$, while it is called \emph{extremal} if the condition $\rket{x}=p\rket{x_1}+(1-p)\rket{x_2}$ implies $\rket{x_1}=\rket{x_2}$. Let now $\ConvH{A}$ denote the convex hull of $\St{A}$. In the case where $\St{A}$ coincides with $\ConvH{A}$ for every system $\sys{A}$, the theory is called \emph{convex}. CT and QT are both convex theories. It might happen that a state is not atomic in $\St{A}$, whereas it is in $\ConvH{A}$.
We will denote by $\Extst{A}$ the set of extremal points of $\ConvH{A}$. Notice that $\Extst{A}$ contains the null state $\rket{\varepsilon}_\sys{A}$.
The event associated to a singleton test---namely a test where the outcome space is a singleton, i.e.~having a single outcome---is called \emph{deterministic}. The interpretation of a deterministic event is that the physical process considered happens with certainty, i.e.~with probability 1. For instance, a state is deterministic if and only if it gives probability 1 on every deterministic effect, or, in other words, if and only if it is normalised. The deterministic extremal states are historically called \emph{pure states}. Finally, deterministic states which are not extremal are
those historically called \emph{mixed states}. More generally,
we will call \emph{mixed} those states which are neither extremal nor atomic.

%
%

It is often convenient to consider the real span of sets of states $\St{A}$, which is a linear space denoted by $\StR{A}\coloneqq \Span_{\mathbb{R}}\St{A}$. The null state $\rket\varepsilon_\sys A$ is represented in $\StR{A}$ by the null vector $0$. Every system $\sys{A}$ is then associated to a quantity $\D{A}\coloneqq \dim \StR{A}$, which is called \emph{the dimension of the system $\sys{A}$}. Since effects are separating for states and viceversa, one has that $\dim\StR{A}=\dim\EffR{A}$. As usual throughout the literature, we consider finite-dimensional OPTs, namely theories where $\D{A}<+\infty$ for all systems $\sys{A}\in\Sys{\Theta}$. The latter assumption means that we are considering systems whose states can be completely probed via the statistics of a finite number of experiments. In CT, for instance, $\D{A}$ is the number of perfectly distinguishable states of a system $\sys{A}$, while in QT one has $\D{A}=d_{\sys{A}}^2$, where $d_{\sys{A}}$ is the dimension of the Hilbert space associated to the system $\sys{A}$.

We conclude the present review considering a fundamental structure of operational theories: \emph{parallel composition}. Indeed, the last piece of information one needs to characterise an OPT is a recipe to form compounds out of systems and events available to local experimenters. We denote parallel composition by the symbol $\boxtimes$, writing, for two arbitrary tests $\T{E}_1=\lbrace \T{E}_x \rbrace_{x\in\rX}$ and $\T{E}_2=\lbrace \T{E}_y \rbrace_{y\in\rY}$, $\T{E}_1\boxtimes\T{E}_2\equiv\lbrace \T{E}_x\boxtimes\T{E}_y \rbrace_{(x,y)\in\rX\times\rY}$. The latter, from a probabilistic point of view, represent uncorrelated tests. The main property of parallel composition is the following:
\begin{align}\label{eq:parallel_commuting}
(\tA\boxtimes\tB)\circ(\tC\boxtimes\tD)=(\tA\circ\tC)\boxtimes(\tB\circ\tD).
\end{align}
Property~\eqref{eq:parallel_commuting} states that the two operations of sequential and parallel composition commute. Also parallel composition distributes over sums. In the case of systems, states, and effects, we will use the following notation: $\sys{AB}\coloneqq\sys{A}\boxtimes\sys{B}$, $\rket{\rho}_{\sys{A}}\rket{\sigma}_{\sys{B}}\coloneqq\rket{\rho}_{\sys{A}}\boxtimes\rket{\sigma}_{\sys{B}}$, and $\rbra{a}_{\sys{A}}\rbra{b}_{\sys{B}}\coloneqq \rbra{a}_{\sys{A}}\boxtimes\rbra{b}_{\sys{B}}$. One has then to specify, for all transformations $\T{E}_x\in\Transf{A}{B}$ and $\T{D}_y\in\Transf{C}{D}$, how the composite event $\T{E}_x\boxtimes\T{D}_y$ embeds into the total space of transformations $\Transf{AC}{BD}$. In QT, for instance, this operation is given by the standard tensor product $\otimes$. Notice that both the sequential and the parallel composition of two deterministic transformations are deterministic. However, in principle, the property of atomicity may be not preserved under both kinds of composition.
\section{Parallel composition and the presence of entangled states in arbitrary Operational Probabilistic Theories}\label{sec:entangl_OPTs}
The existence of parallel composition entails a prescription to assign a dimension $\D{AB}$ to a composite system $\sys{AB}$ as a function of the dimensions $\D{A}$ and $\D{B}$  of the local systems $\sys{A},\sys{B}$. $\St{AB}$ contains at least the parallel composition of the states of $\sys{A}$ and $\sys{B}$, which can be composed independently. By virtue of property~\eqref{eq:parallel_commuting}, for every OPT $\Theta$ one has the following inequality (see e.g.~Ref.~\cite{hardy2012limited}):
\begin{align}\label{eq:dimAB_geq}
\D{AB}\geq\D{A}\D{B}\quad \forall\sys{A},\sys{B}\in\Sys{\Theta}.
\end{align}
This leads us to introduce an ``excess dimension'' of the composite system $\rA\rB$ as follows:
\begin{align*}
	\Delt{AB}\coloneqq\D{AB}-\D{A}\D{B}\quad\forall \rA,\rB\in\Sys{\Theta}.
	\end{align*}
From Eq.~\eqref{eq:dimAB_geq}, one can see  that $\Delt{AB}$ is in general a non-negative integer for all systems $\sys{A},\sys{B}$. In both  CT and QT one has $\Delt{AB}=0$. We then see how this relates to the degree of holism required in the task of state discrimination. 
\begin{property}[$n$-local discriminability~\cite{hardy2012limited}]
	Let $n\leq m$. The effects obtained as a conic combination of the parallel compositions of effects $a_1,a_2,\ldots,a_l$, where $a_j$ is $k_j$-partite with $k_j\leq n$ for all values of $j$, are separating for $m$-partite states.
\end{property}
In the simplest case where $n=1$, the property is called \emph{local discriminability}. In both CT and QT, the parallel compositions of local effects are separating for multipartite states, namely both theories satisfy local discriminability. The next result, whose proof can be found in Ref.~\cite{bookDCP2017}, characterises the composition rules for the dimensions of local systems in theories satisfying local discriminability.
\begin{proposition}\label{prop:local_discriminability}
	Let $\Theta$ be an OPT. Then $\Theta$ satisfies local discriminability if and only if the following rule holds:
	\begin{align}\label{eq:dimAB_equal}
	\D{AB}=\D{A}\D{B}\quad\forall\sys{A},\sys{B}\in\Sys{\Theta}.
	\end{align}
\end{proposition}
As a consequence of the previous result, if a theory $\Theta$ does not satisfy local discriminability, then there exist two systems $\sys{A},\sys{B}\in\Sys{\Theta}$ such that $\Delt{AB} >0$. Another important feature shared by both CT and QT is the following.
\begin{property}[atomicity of state-composition\footnote{One may argue that the property of  atomicity of state-composition should be renamed {\em purity of state-composition}. However, as opposed to the notion of {\em atomicity} (from conic analysis) and {\em extremality} (from convex analysis), the notion of {\em purity} is a historical one introduced only for extremal deterministic states. In theories where every state is proportional to a deterministic one, there is no need for a distinction between atomicity and purity. With a slight abuse of terminology, in several situations, purity is used as a synonym of atomicity.}]
The parallel composition of two atomic states is atomic.
\end{property}

In this work, we relate the properties of local discriminability and atomicity of state-composition to the presence of entangled states. We now recall the definition of separable and entangled states. Given two systems $\sys{A},\sys{B}$, the \emph{separable states} of the bipartite system $\sys{AB}$ are those of the form:
\begin{align}\label{eq:separable_states}
&\rket{\sigma}_{\sys{AB}} = \sum_{i\in I} \rket{\alpha_i}_{\sys{A}}\rket{\beta_i}_{\sys{B}},
\end{align}
with $\St{A}\ni\rket{\alpha_i}_\sys A\neq\rket\varepsilon_\sys A$, $\St{B}\ni\rket{\beta_i}_\sys B\neq\rket\varepsilon_\sys B$ for every $i\in I$.
This set of states contains all states that can be prepared using only Local Operations and Classical Communication (LOCC)---but is generally larger than the set of LOCC states. In CT and QT the converse is also true. By negation, \emph{entangled states} are those states that are \emph{non-separable}.
The two following results hold for arbitrary OPTs.
\begin{proposition}\label{prop:loc_discr}
	Let $\Theta$ be an OPT. If $\Theta$ does not satisfy local discriminability, then it admits of entangled states.
\end{proposition}
\begin{proof}
By hypothesis, from Eqs.~\eqref{eq:dimAB_geq} and~\eqref{eq:dimAB_equal} we have that there exist $\sys{A},\sys{B}\in\Sys{\Theta}$ such that $\D{AB}>\D{A}\D{B}$. Since product states generate a subspace of $\StR{AB}$ of dimension $\D{A}\D{B}$, containing all separable states, $\St{AB}$ must contain at least one state that is not separable. Then $\Theta$ admits of entangled states.
\end{proof}
\begin{proposition}\label{prop:atom_comp}
	Let $\Theta$ be an OPT. If $\Theta$ does not satisfy atomicity of state-composition, then it admits of entangled states.
\end{proposition}
\begin{proof}
	By hypothesis, there exist two systems $\sys{A},\sys{B}$ and two atomic states $\rket{\rho}_{\sys{A}},\rket{\tau}_{\sys{B}}$ whose parallel composition $\rket{\sigma}_{\sys{AB}}=\rket{\rho}_{\sys{A}}\rket{\tau}_{\sys{B}}$ is not atomic. By contradiction, suppose that the states of the composite system $\sys{AB}$ are all separable, namely of the form~\eqref{eq:separable_states}. 
	By hypothesis, $\rket{\sigma}_{\sys{AB}}$ is not atomic, and then it admits a decomposition of the form
	\begin{align}
	\rket{\sigma}_{\sys{AB}} = \sum_{i=1}^p\rket{\alpha_i}_{\sys{A}}\rket{\beta_i}_{\sys{B}},\label{eq:decompabsurd}
	\end{align}
%
	where there are two indices $1\leq l<k\leq p$ such that $\rket{\alpha_l}_{\sys{A}}\rket{\beta_l}_{\sys{B}}\not\propto\rket{\alpha_k}_{\sys{A}}\rket{\beta_k}_{\sys{B}}$. This implies that $\rket{\alpha_l}_{\sys{A}}\not\propto\rket{\alpha_k}_{\sys{A}}$ or $\rket{\beta_l}_{\sys{B}}\not\propto\rket{\beta_k}_{\sys{B}}$---say the former without loss of generality. Now, since the decomposition in Eq.~\eqref{eq:decompabsurd} can be taken such that $\rket{\beta_i}_\sys B\neq\rket\varepsilon_\sys B$ for every $1\leq i\leq p$ without loss of generality, then for every $1\leq i\leq p$ there must exist a deterministic effect $\rbra{\tilde e_i}_\sys B$ such that $\rbraket{\tilde e_i}{\beta_i}_\sys B>0$ .
	Let us now choose one of the deterministic effects $\rbra{\tilde{e}_j}_\sys{B}$. Let also $\rket{\tilde\tau}_\sys B$ be such that $\{\rket\tau_\sys B,\rket{\tilde\tau}_\sys B\}$ is a test.	Then 
	\begin{align*}
		&\rbra{\tilde{e}_j}_\sys{B}\rket{\sigma}_{\sys{AB}}+\rbraket{\tilde e_j}{\tilde\tau}_\sys B\rket\rho_\sys A=\\
		&=\rbraket{\tilde{e}_j}{\tau}_{\sys{B}}\rket{\rho}_{\sys{A}}+\rbraket{\tilde e_j}{\tilde\tau}_\sys B\rket\rho_\sys A=\rket\rho_\sys A=\\
		&=\sum_{i=1}^p\rbraket{\tilde{e}_j}{\beta_i}_{\sys{B}}\rket{\alpha_i}_{\sys{A}}+\rbraket{\tilde e_j}{\tilde\tau}_\sys B\rket\rho_\sys A.
		\label{eq:rho_atomic}
	\end{align*}
For a fixed choice of $j$, the term $\rbraket{\tilde e_j}{\beta_j}_\sys B\rket{\alpha_j}_\sys A$ is non-null, and by the atomicity of $\rket\rho_\sys A$, one must then have $\rket{\alpha_j}_\sys A\propto\rket\rho_\sys A$. Since this is true for every $1\leq j\leq p$, we come to a contradiction with the hypothesis that $\rket{\alpha_l}_\sys A\not\propto\rket{\alpha_k}_\sys A$.
This implies that there exists a state of $\sys{AB}$ which is not separable, namely $\Theta$ admits of entangled states.
\end{proof}
A counterexample to the converse of both Propositions~\ref{prop:loc_discr} and~\ref{prop:atom_comp} is simply given by QT, which contains entangled states and also satisfies both local discriminability and atomicity of state-composition.

\section{Entanglement in simplicial theories}\label{sec:entangl_simpl}
In the present section, we characterise simplicial theories from the point of view of the admissibility of entangled states, finding that the converse of Proposition~\ref{prop:loc_discr} holds in general for this family of theories, whereas the converse of Proposition~\ref{prop:atom_comp} holds in the case of simplicial theories satisfying $n$-local discriminability for some positive integer $n$.

\begin{definition}[simplicial theory]
	A \emph{simplicial theory} $\Theta$ is a finite-dimensional OPT where the extremal states of every system $\sys{A}\in\Sys{\Theta}$ are the vertices of a $\D{A}$-simplex.
\end{definition}
Notice that a $\D{A}$-simplex is the convex hull of $\D{A}+1$ vertices, which in the present context are the elements of $\Extst{A}$ (including also the null state $\rket{\varepsilon}_\sys{A}$). The property of simpliciality for an OPT $\Theta$ implies that, for every system $\sys{A}\in\Sys{\Theta}$: (i) $\sys{A}$ has exactly $\D{A}$ non-null extremal states, and (ii) every state of $\sys{A}$ has a unique decomposition as a convex combination of non-null extremal states. We will make extensive use of the two previous properties in the remainder of the paper.
\begin{property}[joint perfect discriminability]\label{prope:jpd}
	Let $\Theta$ be an OPT and $\sys{A}\in\Sys{\Theta}$. A set of states $\{\rket{\rho_i}_\sys A\}_{i=1}^n$ is jointly perfectly discriminable if there exists an observation test $\{ \rbra{a_i}_{\sys{A}}\}_{i=1}^n$ such that:
	\begin{align*}
	\rbraket{a_i}{\rho_j}_\sys{A} = \delta_{ij}\quad\forall i,j\in\lbrace 1,2,\ldots n\rbrace.
	\end{align*}
\end{property}
As it has been recalled, \emph{classical theories} are those simplicial theories where the pure states satisfy property~\ref{prope:jpd}. On the other hand, CT is a simplicial theory satisfying convexity, local discriminability, and joint perfect discriminability of pure states for every system. Further, the following property, i.e.~\emph{causality}, is usually also assumed for CT.
\begin{property}[causality]\label{prope:causality}
	An OPT is causal if every system admits of a unique deterministic effect.
\end{property}
The above definition of causality is in fact equivalently formulated in a more intuitive way, as follows: A theory is causal if the probabilities of the outcomes of preparations are independent of the choice of the observations connected at their output. The latter statement, indeed, can be recognised as the most popular formulation of the notion of causality in a physical theory. However, its content is equivalent to property~\ref{prope:causality} (for a proof, see Ref.~\cite{bookDCP2017}). We now prove that the property of causality is in fact intrinsic to that of simpliciality, thus establishing that all simplicial theories---CT included---are indeed inherently causal.
\begin{theorem}\label{thm:causality}
	Simplicial theories are causal.
\end{theorem}
\begin{proof}
	Let $\Theta$ be a simplicial theory, $\sys{A}\in\Sys{\Theta}$ and $\rket{i}_{\sys{A}}$ an extremal and non-null state. Suppose that $\rket{i}_{\sys{A}}$ is not deterministic: then it can be completed to a deterministic one, namely there exists a non-null state $\rket{\sigma}_{\sys{A}}$ such that $\rket{\rho}_{\sys{A}}=\rket{i}_{\sys{A}}+\rket{\sigma}_{\sys{A}}$ is deterministic. Now, let $\rbra{b_j}_{\sys{A}}$ for $j=1,\ldots,\D{A}$ be the linear functionals such that $\rbraket{b_j}{j'}_{\sys{A}}=\delta_{jj'}$ for every non-null extremal state $\rket{j'}_{\sys{A}}$: this implies, by simpliciality, that $0\leq \rbraket{b_j}{\tau}_{\sys{A}} \leq 1$ for all $\rket{\tau}_{\sys{A}}\in\St{A}$ and all $j$. One then has:
	\begin{align*}
		\rbraket{b_i}{\rho}_{\sys{A}} = \rbraket{b_i}{i}_{\sys{A}}+\rbraket{b_i}{\sigma}_{\sys{A}}= 1+\rbraket{b_i}{\sigma}_{\sys{A}}\leq 1.
	\end{align*}
	namely	$\rbraket{b_i}{\sigma} =0$ and $\rbraket{b_i}{\rho}=1$. Then, posing $\Extst{A}\coloneqq\lbrace \rket{k}_{\sys{A}} \rbrace_{k=0}^{\D{A}}$, with $\rket0_\sys A\coloneqq\rket\varepsilon_\sys A$ one has the following:
	\begin{equation*}
	\rket{\sigma}_{\sys{A}}=\sum_{k=0}^{\D{A}}p_k\rket{k}_{\sys{A}},\qquad p_k\geq 0\ \forall k,\ p_i=0,\ \sum_{k=0}^{\D{A}}p_k\leq 1.
	\end{equation*}
Let us now pose
		\begin{equation*}
\rbra{e}_{\sys{A}}\coloneqq\sum_{j=1}^{\D{A}}\rbra{b_j}_{\sys{A}}.
		\end{equation*}
		Clearly $0\leq \rbraket{e}{\tau}_{\sys{A}} \leq 1$ 
		for all $\rket{\tau}_{\sys{A}}\in\St{A}$, and $\rbraket{e}{i}_{\sys{A}}=1$, implying
	$\rbraket{e}{\sigma}_{\sys{A}}=0$, i.e.~$p_k=0\ \forall k\neq 0$. This shows that $\rket{i}$ is deterministic. Being independent of $i$, the above argument proves that all non-null extremal states are deterministic. Extending then the above argument to arbitrary convex combinations of non-null extremal states of the form $\sum_{k=0}^{\D{A}}p_k\rket{k}_{\sys{A}}$, it is easy to see that a state is deterministic if and only if $\sum_{k=1}^{\D{A}}p_k=1$. In other words, the set of deterministic states coincides with the convex hull of extremal non-null states. Now, the effect $\rbra{e}_\rA$ amounts to unit on all deterministic states, hence it is deterministic. Since the non-null extremal states of a simplicial theory are complete and linearly independent, there exists a unique effect $\rbra{\tilde{e}}_{\sys{A}}$ such that $\rbraket{\tilde{e}}{j}_\rA=1$ for all extremal non-null states $\rket{j}_{\sys{A}}$. Thus the  deterministic effect $\rbra{e}_{\sys{A}}$ is unique.
\end{proof}
Now we show that, for any simplicial theory, the converse of Proposition~\ref{prop:loc_discr} is also true.
\begin{theorem}\label{thm:loc_discr_simplicial}
	Let $\Theta$ be a simplicial theory. Then $\Theta$ admits of entangled states if and only if it does not satisfy local discriminability.
\end{theorem}
\begin{proof}
($\Leftarrow$) The implication holds true by Proposition~\ref{prop:loc_discr}.
($\Rightarrow$) By hypothesis, the theory $\Theta$ is simplicial.
Let us denote the non-null extremal states of a system $\sys{X}$ by $\rket{l}_{\sys{X}}$, for $l=1,\ldots,\D{X}$. As a straightforward consequence of Theorem~\ref{thm:causality}, such states are also deterministic. Then, for all  $\sys{A},\sys{B}\in\Sys{\Theta}$ and all non-null $\rket{i}_{\sys{A}}\in\Extst{A},\rket{j}_{\sys{B}}\in\Extst{B}$, there exists 
a non-empty set $I_{ij}\subseteq\lbrace1,\ldots,\D{AB}\rbrace$ such that 
\begin{align}\label{eq:decomposition_pure}
	\rket{i}_{\sys{A}}\rket{j}_{\sys{B}}=
	\sum_{k\in I_{ij}}p_{k}^{ij}\rket{k}_{\sys{AB}},
\end{align}
with $p^{ij}_k>0$ for $k\in I_{ij}$ and $\sum_{k\in I_{ij}}p^{ij}_k=1$.
Let $\rbra{e}_{\sys{A}}\in\Eff{A}$ be the (unique) deterministic effect of $\sys{A}$.
Since for all $k\in\lbrace 1,\ldots,\D{AB}\rbrace$ $\rbra{e}_{\sys{A}}\rket{k}_{\sys{AB}}\in\St{B}$ is deterministic, it is also non-vanishing for all $k$.
Moreover, let us suppose that there exist two different pairs of indices $(\tilde{i},\tilde{j}),(\tilde{i}',\tilde{j}')$ such that $\tilde{J}\coloneqq I_{\tilde{i}\tilde{j}}\cap I_{\tilde{i}'\tilde{j}'}\neq \emptyset$. Without loss of generality, we can assume $\tilde{j}\neq\tilde{j}'$. Choose $\tilde{k}\in\tilde{J}$: being the convex decomposition into non-null extremal states unique by simpliciality, from Eq.~\eqref{eq:decomposition_pure} we conclude that $\rbra{e}_{\sys{A}}\rket{\tilde{k}}_{\sys{AB}}=\rket{\tilde{j}}_{\sys{B}}=\rket{\tilde{j}'}_{\sys{B}}$,
which is absurd. Thus $I_{ij}\cap I_{i'j'}= \emptyset$ for every different pair of indices. 
Now, 
let us suppose that $\Theta$ satisfies local discriminability. Then, by Proposition~\ref{prop:local_discriminability}, Eq.~\eqref{eq:dimAB_equal} holds, and $I_{ij}\subseteq\{1,\ldots,\D A\D B\}, \forall i,j$. By conditions $I_{ij}\neq\emptyset$ and $I_{ij}\cap I_{i'j'}=\emptyset$, every set $I_{ij}$ must be a singleton. This implies that for every pair $i,j$ there exists $\rket k_{\rA\rB}$ such that $\rket k_{\rA\rB}=\rket i_\rA\rket j_\rB$, namely the parallel composition of non-null extremal states is a non-null extremal state.
Accordingly, Eq.~\eqref{eq:dimAB_equal}
implies that the states of the form $\rket i_\rA\rket j_\rB$ exhaust the set of non-null extremal states of the composite system $\rA\rB$. As a consequence, $\Theta$ does not admit of entangled states.
\end{proof}
Next, we prove that, for a simplicial theory satisfying $n$-local discriminability for some positive integer $n$, the converse of Proposition~\ref{prop:atom_comp} is also true. In order to do this, we need the following result, that is proved in Appendix~\ref{app:A}.
\begin{proposition}\label{puropuro} Let $\Theta$ be a simplicial OPT satisfying $n$-local discriminability for some positive integer $n$. For every pair of systems $\rA,\rB\in\Sys{\Theta}$, every non-null extremal state $\rket{\lambda}_{\rA\rB}\in\Extst{AB}$ convexly refines the parallel composition of some pair of pure states $\rket{\rho}_{\sys{A}}\rket{\sigma}_{\sys{B}}\in\St{AB}$.
\end{proposition}
We are now in position to prove the following theorem.
\begin{theorem}\label{thm:atomicity}
Let $\Theta$ be a simplicial theory satisfying $n$-local discriminability for some positive integer  $n$. Then $\Theta$ admits of entangled states if and only if it does not satisfy atomicity of state-composition.
\end{theorem}
\begin{proof}
($\Leftarrow$) The implication holds true by Proposition~\ref{prop:atom_comp}. 
($\Rightarrow$) According to Proposition~\ref{puropuro}, under the hypothesis of the theorem every non-null extremal state of a bipartite system refines some parallel composition of extremal states. Since by hypothesis the theory admits of entangled states, by Theorem \ref{thm:loc_discr_simplicial} and Proposition \ref{prop:local_discriminability} there exists a pair of systems $\rA,\rB\in\Sys{\Theta}$ such that $D_{\rA\rB}>D_{\rA}D_{\rB}$. Then, there are more non-null extremal states of the composite system than parallel compositions of (non-null) extremal states. Thus, there must be a parallel composition of extremal states $\rket{i}_\rA\rket{j}_\rB$ that is refined by more than one non-null extremal states, i.e. $\rket{i}_\rA\rket{j}_\rB=\sum_{k\in I_{ij}}p^{ij}_k\rket{k}_{\rA\rB}$, namely atomicity of state-composition does not hold. 
\end{proof}
The following lemma provides a characterisation of the parallel composition of states in simplicial theories satisfying $n$-local discriminability for some positive integer $n$.
\begin{lemma}\label{lem:refinement}
	Let $\Theta$ be a simplicial theory satisfying $n$-local discriminability for some positive integer $n$. Then, for all systems $\sys{A},\sys{B}\in \Sys{\Theta}$ and non-null extremal states $\rket{k}_{\sys{AB}}\in\Extst{AB}$, there exists a unique product of non-null extremal states $\rket{i_kj_k}_{\sys{AB}}=\rket{i_k}_{\sys{A}}\rket{j_k}_{\sys{B}}$ such that $\rket{k}_{\sys{AB}}$ convexly refines $\rket{i_kj_k}_{\sys{AB}}$.
\end{lemma}
\begin{proof}
Existence is provided by Proposition~\ref{puropuro}, and uniqueness has been proved in the proof of Theorem \ref{thm:loc_discr_simplicial}.
\end{proof}
Notice that the map $k\mapsto i_kj_k$ of the above lemma is not injective in general. In particular, the map is not injective as long as the theory does not satisfy atomicity of state-composition. As a straightforward consequence of Lemma~\ref{lem:refinement}, we are now in position to provide a general classification of the composite state-spaces in simplicial theories satisfying $n$-local discriminability for some integer $n$.
\begin{theorem}[classification of composite state-spaces in simplicial theories]\label{thm:classification}
	Let $\Theta$ be a simplicial theory satisfying $n$-local discriminability for some integer $n$. For every pair of systems $\sys{A},\sys{B}\in\Sys{\Theta}$, the state-space $\St{AB}$ is classified by the following: (i) a choice of the dimension of the composite system $\D{AB}$; (ii) an unambiguous labeling for the pure states of $\sys{AB}$ as $\rket{(ij)_k}_{\sys{AB}}$---for $i\in\lbrace 1,\ldots,\D{A}\rbrace$, $j\in\lbrace 1,\ldots,\D{B}\rbrace$, and $k$ in finite sets $I_{ij}$---and a choice of probability distributions $p_{k}^{ij}$, such that:
		\begin{align}\label{eq:classification_states}
			\rket{i}_{\sys{A}}\rket{j}_{\sys{B}} = \sum_{k\in I_{ij}} p_{k}^{ij}\rket{(ij)_k}_{\sys{AB}}
		\end{align}
		with $p_{k}^{ij}>0$ and $\sum_{k\in I_{ij}} p_{k}^{ij}=1$ for all non-null extremal states  $\rket{i}_{\sys{A}}\in\Extst{A},\rket{j}_{\sys{B}}\in\Extst{B}$.
\end{theorem}		
Notice that not all the choices of the coefficients $p_k^{ij}$ in Eq.~\eqref{eq:classification_states} necessarily lead to a consistent theory. However, there are cases of classical theories beyond CT, proving that a consistent choice of $p_k^{ij}$ is possible~\cite{DEP2020}.

We conclude our investigation on simplicial theories proving that these theories, whereas they may admit of entangled states, cannot admit of superposition (in an operational sense), or purification (under the hypothesis of $n$-local discriminability for some positive integer $n$). We formulate the superposition principle in three possible ways.
\begin{property}[superposition principle]\label{prope:superposition}
	Let $\mathcal{D}=\lbrace \rket{\rho_i}_{\sys{A}} \rbrace_{i=1}^d$, with $d\geq 2$, be any maximal set of jointly perfectly discriminable pure states of a system $\sys{A}$. Let
	$\mathscr{A}=\lbrace \rbra{a_i} \rbrace_{i=1}^d$ denote an observation such that $\rbraket{a_i}{\rho_j}=\delta_{ij}$ for all $i,j\in\lbrace 1,\ldots, d \rbrace$, $\mathbf{p}=\lbrace p_i \rbrace_{i=1}^d$
	a probability distribution, and $\rket{\sigma}_\sys{A}$ a pure state of $\sys A$.
	\begin{itemize}
		\item[(i)] \emph{Ultraweak:} For every choice of
		$\mathbf{p}$, there exist $\mathscr{A}$ and $\rket{\sigma}_\sys{A}$ such that the following holds:
		\begin{align}\label{eq:superposition}
		p_i=\rbraket{a_i}{\sigma}_{\sys{A}}\quad \forall i=1,\ldots,d.
		\end{align}
		\item[(ii)] \emph{Weak:} For every choice of
		$\mathbf{p}$ and $\mathscr{A}$, there exists $\rket{\sigma}_\sys{A}$ such that Eq.~\eqref{eq:superposition} holds.
		\item[(iii)] \emph{Strong:} For every choice of
		$\mathbf{p}$, there exists $\rket{\sigma}_\sys{A}$ such that, for every choice of $\mathscr{A}$, Eq.~\eqref{eq:superposition} holds.
	\end{itemize}
\end{property}
\begin{property}[purication principle]
	Let $\rket{\rho}_\sys{A}$ be a deterministic state. Then there exists a system $\sys{B}$, a pure state $\rket{\Sigma_\rho}_{\sys{AB}}$, and a deterministic effect $\rbra{e_\rho}_{\sys{B}}$, such that:
	\begin{align*}
	\rket{\rho}_\sys{A}=\rbra{e_\rho}_\sys{B}\rket{\Sigma_\rho}_{\sys{AB}}.
	\end{align*}
\end{property}
Both the ultraweak superposition and the purification principles are satisfied by QT---but, clearly, not by CT. Interestingly, a simplicial or classical theory where the states of every system are exhausted by the pure ones---e.g.~this is the case of Deterministic Classical Theory~\cite{d2019information}---trivially satisfies the purification principle, since all states are pure.
\begin{theorem}[no superposition for simplicial theories]\label{thm:no_superposition}
	Let $\Theta$ be an OPT. If $\Theta$ is simplicial, then there is no system in $\Theta$ satisfying the weak formulation of the superposition principle. If $\Theta$ is classical, then there is no system in $\Theta$ satisfying any formulation of the superposition principle.
\end{theorem}
\begin{proof}
Let $\Theta$ be a simplicial theory, and $\sys A$ a system of $\Theta$. First, we notice that, if there does not exist a set $\mathcal{D}=\lbrace \rket{i}_{\sys{A}} \rbrace_{i=1}^d$ of jointly perfectly discriminable pure states with $d\geq 2$, property~\ref{prope:superposition} is empty. Accordingly, we assume that $\mathcal{D}$ exists, that it is maximal, and, as a necessary condition, that $\D{A}\geq 2$.  Let us pose $I\coloneqq\lbrace 1,\ldots,d \rbrace$ and $K\coloneqq\lbrace d+1,\ldots,\D{A} \rbrace$. Let $\lbrace \rbra{f_l}_{\sys{A}} \rbrace_{l\in I\cup K}$ be the set of linear functionals such that $\rbraket{f_l}{j}_\sys{A}=\delta_{lj}$ for all non-null extremal states $\rket{j}_\sys{A}\in \Extst{A}$ and all $l\in I\cup K$. Then, every observation $\T{A}=\lbrace \rbra{a_i} \rbrace_{i\in I}$ such that $\rbraket{a_i}{i'}=\delta_{ii'}$ for all $i,i'\in I$ has the following form:
\begin{align}\label{eq:perf_discr_obs}
\rbra{a_i}_\sys{A} = \rbra{f_i}_{\sys{A}} + \sum_{k\in K} q_k^i \rbra{f_k}_{\sys{A}},
\end{align}
with $q_k^i\geq 0$ and $\sum_{i'\in I} q_k^{i'}=1$ for all $i\in I,k\in K$, since it must be $\sum_{i'\in I}\rbra{a_{i'}}_\sys{A} = \rbra{e}_\sys{A}\equiv\sum_{l\in I\cup K}\rbra{f_{l}}_\sys{A}$. If the theory $\Theta$ is not classical, then $|I|\equiv d< \D{A}$ and $|K|\geq 1$. Being the number of pure states finite, for every choice of $\T{A}=\lbrace \rbra{a_i} \rbrace_{i\in I}$ of the form~\eqref{eq:perf_discr_obs}, Eq.~\eqref{eq:superposition} might be satisfied only for a finite number of choices of probability distributions $\mathbf{p}=\lbrace p_i \rbrace_{i\in I}$, namely those with  $p_i=\delta_{ii_0}$ and $i_0\in I$, or those with $p_i=q^i_k$ for $k\in K$. Therefore no simplicial theory satisfies the weak formulation (ii) of property~\ref{prope:superposition}. If the theory $\Theta$ is classical, then $|I|\equiv d=\D{A}$ and $K=\emptyset$ by definition. This means that an observation $\T{A}$ of the form~\eqref{eq:perf_discr_obs} is unique---in particular, $\T{A}=\lbrace \rbra{f_i} \rbrace_{i\in I}$. Accordingly, Eq.~\eqref{eq:superposition} can be satisfied for a unique (modulo permutations of the indices) choice of probability distribution $\mathbf{p}=\lbrace p_i \rbrace_{i=1}^d$, namely $p_i=\delta_{ii_0}$ for some $i_0\in I$ and all $i\in I\setminus\{i_0\}$. Therefore no classical theory satisfies the ultraweak formulation (i) of property~\ref{prope:superposition}.
%
\end{proof}
\begin{theorem}[no purification for simplicial theories with $n$-local discriminability]\label{thm:no_purification}
	Let $\Theta$ be a simplicial theory satisfying $n$-local discriminability for some $n$. Then, there is no mixed state in $\Theta$ having a purification.
\end{theorem}
\begin{proof}
	By Theorem~\ref{thm:classification}, every pure state of a composite system convexly refines the parallel composition $\rket i_\sys A\rket j_\sys B$ of two extremal states. In particular, denoting the (unique) deterministic effect of $\sys{B}$ by $\rket{e}_{\sys{B}}$, by direct inspection of Eq.~\eqref{eq:classification_states} one concludes that, for any pure state $\rket{(ij)_k}_{\sys{AB}}$, it must be $\rbra{e}_{\sys{B}}\rket{(ij)_k}_{\sys{AB}}=\rket{i}_{\sys{A}}$, since the convex decomposition into non-null extremal states is unique by simpliciality. Accordingly, any marginal state of a pure state is pure, implying that no mixed stated can be purified in $\Theta$.
\end{proof}

\section{Discussion and conclusions}\label{sec:conclusions}
In Propositions~\ref{prop:loc_discr} and~\ref{prop:atom_comp} we showed that, in arbitrary probabilistic theories, the presence of entangled states is a consequence of the failure of local discriminability or, independently, of atomicity of state-composition. Then we specialised to simplicial theories, proving the converse of Propositions~\ref{prop:loc_discr} and~\ref{prop:atom_comp} for this class of theories. In particular, Theorem~\ref{thm:loc_discr_simplicial} asserts that a simplicial theory admits of entangled states if and only if it does not satisfy local discriminability. Theorem~\ref{thm:atomicity} and Theorem~\ref{thm:classification} are proved  under the  hypothesis of $n$-local discriminability for some integer $n$. On the one hand, a simplicial theory with $n$-local discriminability contains entangled states if and only if it does not satisfy atomicity of state-composition. This implies that there exists at least a pair of systems $\sys{A},\sys{B}$ such that at least a product of pure states of $\sys{A}$ and $\sys{B}$ is not pure. On the other hand, Theorem~\ref{thm:classification} provides a simple classification of the parallel composition rule for such theories, as a consequence of the failure of atomicity of composition. Eq.~\eqref{eq:classification_states}, in turn, implies that the marginal state of a pure entangled state in a simplicial theory with $n$-local discriminability is always pure.

As far as classical theories are concerned, on the one hand, the above results entail that the usual notion of classicality does not prevent the presence of entanglement in a physical theory. Indeed, CT is the only classical theory without entanglement.\footnote{More precisely, CT is the only classical theory that is: (i) without entanglement, (ii) convex, and (iii) such that every transformation which is compatible with simpliciality, convexity, and local discriminability is allowed.}
On the other hand, Theorem~\ref{thm:causality} asserts that the notion of classicality given by the simplicial structure of states constrains a theory to be necessarily causal. Interestingly, in the above sense, this means that the notion of causality is inherent to that of classicality, representing a no-go theorem for conceiving \emph{non-causal} classical theories. 

Finally, Theorems~\ref{thm:no_superposition} and~\ref{thm:no_purification} highlight two relevant properties that every simplicial theory shares with CT: superposition is not admitted, and no mixed state has a purification if $n$-local discriminability holds for some $n$. Interestingly, at this stage it cannot be excluded that there might exist non-classical simplicial theories satisfying the ultraweak operational formulation of the superposition principle (see property~\ref{prope:superposition}ì). Besides, we notice that Theorem~\ref{thm:no_purification} has a different content from the no-go theorem proven in Ref.~\cite{winczewski2018no}, although the latter holds for convex discrete theories---which strictly include the simplicial ones. Indeed, the no-go theorem of Ref.~\cite{winczewski2018no} states that, in a convex discrete theory, for every system \emph{only a finite number of (mixed) states may possibly have a purification}. However, our Theorem~\ref{thm:no_purification} generalises the latter in the case of simplicial theories with $n$-local discriminability, stating that \emph{not a single (mixed) state admits of a purification}.

The general results of the present work do not rely on any additional structure---such as, convexity, the no-restriction hypothesis~\cite{bookDCP2017,PhysRevA.87.052131}, or other properties---beyond the simplicial one. If a theory is simplicial, there is no complementarity---and it is thus impossible to violate Bell's Inequalities.\footnote{Indeed, if there is no complementarity, there exists a joint probability distribution for the outcomes of any pair of measurements, which then excludes a violation of any probability bound.} Moreover, in presence of local discriminability, the converse is also true~\cite{PhysRevA.94.042108}. While the non-violation of any probability bound means that \emph{every single correlation} of the theory can be described by a local realistic hidden-variable model, this in principle does not imply that there exists a \emph{coherent and complete ontological model}~\cite{s2019noncontextuality} describing \emph{the theory as a whole} in a local realistic fashion. Then, an interesting open question is: does there exist a convenient ontological model for such theories?

All simplicial theories with entanglement contain states with non-null discord (see Ref.~\cite{PhysRevLett.108.120502}), despite being simplicial. Furthermore, as shown in Ref.~\cite{d2019information}, if a theory satisfies the \emph{full-information without disturbance principle}, then the pure states of every system are jointly perfectly discriminable. It follows that in non-classical simplicial theories it is impossible to extract \emph{all} the information without disturbing the measured system. Moreover, simplicial theories with entanglement feature  \emph{hypersignaling}~\cite{PhysRevLett.119.020401}, since for such theories one has $\D{AB}>\D{A}\D{B}$ for at least a pair of systems $\rA,\rB$. At this stage, it is not possible to compare simplicial theories exhibiting entanglement with results on \emph{broadcasting}~\cite{PhysRevLett.99.240501} or \emph{teleportation}~\cite{barnum2012teleportation} in arbitrary probabilistic theories, since local discriminability was therein assumed.

It is interesting to notice that all classical theories have no \emph{dimension mismatch}~\cite{Brunner_2014}, while in the case of simplicial theories the problem is open.  Interestingly, simplicial theories satisfy both Information Causality (IC)~\cite{Pawowski:2009aa} and the Information Content Principle (ICP)~\cite{PhysRevA.95.022119}. This is a consequence of the fact that, in a simplicial theory, the mixing entropy satisfies the sufficient properties guaranteeing that IC and ICP hold in a general simplicial theory~\cite{Barnum2010,PhysRevA.84.042323,PhysRevA.95.022119}. In particular, IC and ICP cannot single out CT among arbitrary classical theories.

More generally, are there device-independent principles~\cite{scarani2012device,GRINBAUM201722} ruling out simplicial theories with entanglement? An explicit construction of a simplicial theory with entanglement, complete with the set of transformations, would allow one to study the simplicial scenario as far as, for example, communication complexity~\cite{PhysRevLett.96.250401} is concerned. One could also examine the relation between simplicial theories with entanglement and the so-called \emph{epistemically restricted} (or \emph{``epirestricted''}) theories---such as Spekkens toy theory~\cite{PhysRevA.75.032110}---or consider them in the light of the literature on the notion of classicality (e.g.~see Refs.~\cite{Navascues2016,PhysRevLett.119.080503,s2019noncontextuality}). We conclude pointing out that our results highlight the relevance of the notion of compositionality~\cite{doi:10.1080/00107510903257624} in the scope of probabilistic theories. Indeed, our results strongly rely on the general properties of the parallel composition of systems in a theory. It would be interesting to investigate the applicability of the methods exploited to broader contexts, e.g.~extending them to theories where the sets of states are quantum, or even to post-quantum theories.

\section*{Acknowledgments}
This publication was made possible through the support of a grant from the John Templeton Foundation, ID \# 60609 ``Quantum Causal Structures''. The opinions expressed in this publication are those of the authors and do not necessarily reflect the views of the John Templeton Foundation.

\bibliography{opt_bib2}

\appendix\section{Proof of Proposition~\ref{puropuro}}\label{app:A}

We here recall that in the main text we proved that: (i) any simplicial theory is causal (see Theorem~\ref{thm:causality}), and (ii) the non-null extremal states of a simplicial theory are deterministic (see proof of Theorem~\ref{thm:causality}). We will make use of the above results in the present appendix. We also introduce the definition of separable state for an arbitrary number of systems, which will be of crucial relevance in the remainder. Let $\sys S=\sys{S_1}\sys{S_2}\cdots\sys{S_{n}}\in\Sys{\Theta}$, and $\rket\rho\in\St{S}$. We say that $\rket\rho$ is {\em separable} if there exist (finitely many) disjoint non-trivial bipartitions $\sys{S}^{a}_0\coloneqq\{i^{a}_1,i^{a}_2,\ldots,\i^{a}_k\}$, $\sys{S}^{a}_1\coloneqq\{j^{a}_1,j^{a}_2,\ldots,j^{a}_{N-k}\}$ of $\{1,2,\ldots,n\}$ such that
\begin{align*}
\rket\rho=\sum_{a\in\mathsf A}p_a\rket{\sigma}_{\sys{S}^a_0}\rket{\tau}_{\sys S^a_1},
\end{align*}
with $p_a>0$ for all $a\in\mathsf A$.

\begin{lemma}\label{lem:separable}
	Let $\Theta$ be a simplicial OPT satisfying $n$-local discriminability for some positive integer $n$. For all $(n+1)$-partite system $\sys S=\sys{S_1}\sys{S_2}\cdots\sys{S_{n+1}}\in\Sys{\Theta}$, every state $\rket{\rho}\in\St{S_1S_2\cdots S_{n+1}}$ admits of a convex decomposition into states each of which convexly refines some separable state of $\sys{S_1S_2\cdots S_{n+1}}$.
\end{lemma}
\begin{proof}
	Take the subset $\mathcal{E}\subseteq\Extst{S_1\cdots S_{n+1}}$ of all non-null extremal states of $\St{S_1\cdots S_{n+1}}$ which convexly refine some separable state. Since $\Theta$ satisfies $n$-local discriminability, this is a spanning set for the space $\StR{S_1\cdots S_{n+1}}$. Moreover, since $\Theta$ is simplicial, the elements of $\mathcal{E}$ are linearly independent. As a consequence, the dimension of $\StR{S_1\cdots S_{n+1}}$ amounts to the cardinality of $\mathcal{E}$, which is then, by simpliciality, a complete set of states convexly generating every state $\rket{\rho}\in \St{S_1\cdots S_{n+1}}$. Equivalently, $\Extst{S_1\cdots S_{n+1}}=\{\rket{\varepsilon}\}\cup\mathcal{E}$.
\end{proof}
\begin{lemma}\label{lem:marginal}
	Let $\Theta$ be a simplicial OPT. Let $\rket{\pi}\in \St{S_1\cdots S_n}$ with $n\geq 2$, so that $\rket{\pi}= \rket{\pi_I}\rket{\pi_J}$, for some states $\rket{\pi_I}\in\sys{S_{I}},\rket{\pi_J}\in\sys{S_{J}}$ with $I\cup J=\{1,\ldots ,n\}$, $I,J\neq\emptyset$, $I\cap J=\emptyset$, and $\sys{S_K}=\sys{S_{k_1}S_{k_2}\cdots S_{k_l}}$ for every $l$-tuple $K\subseteq\{1,\ldots ,n\}$.		
	Let $\rket{\phi}\in \Extst{S_1\cdots S_n}$ be a non-null extremal state that convexly refines $\rket{\pi}$. Finally, given $i\in I$ and $j\in J$, let $\rbra{e_{K_{ij}}}$ denote the deterministic effect on $\sys{S_{K_{ij}}}$ with $K_{ij}=\{1,\ldots,n\}\setminus\{i,j\}$. Then $\rbraket{{e}_{K_{ij}}}{\pi}$ is a product state, and $\rbraket{{e}_{K_{ij}}}{\phi}\in\St{S_iS_j}$ is a physical state that convexly refines $\rbraket{{e}_{K_{ij}}}{\pi}$.
\end{lemma}
\begin{proof}
	The case $n=2$ is trivially true, and we will then assume $n\geq 3$ in the following. By hypothesis, we can pose $\rket{\pi}=p\rket{\phi}+(1-p)\rket{\sigma}$, where $\rket{\sigma}$ is a deterministic state and $p\in(0,1]$. By construction we have that 
	\begin{align}\label{eq:marginal}
	\rbraket{{e}_{K_{ij}}}{\pi}= p\rbraket{{e}_{K_{ij}}}{\phi}+(1-p)\rbraket{{e}_{K_{ij}}}{\sigma}\in\St{S_iS_j},
	\end{align}
	and clearly both $\rbraket{{e}_{K_{ij}}}{\phi}$ and $\rbraket{{e}_{K_{ij}}}{\sigma}$ are deterministic states of $\sys{S_iS_j}$. Moreover, by causality, $\rbra{e_{K_{ij}}}=\rbra{e_{I\setminus\{i\}}}\rbra{e_{J\setminus\{j\}}}$, then $\rbraket{{e}_{K_{ij}}}{\pi}$ is a product state of $\sys{S_iS_j}$. Since the convex decomposition into non-null extremal states is unique by simpliciality, from Eq.~\eqref{eq:marginal} we conclude that $\rbraket{{e}_{K_{ij}}}{\phi}$---that is non-null, although it may possibly be non-extremal---convexly refines $\rbraket{{e}_{K_{ij}}}{\pi}$.
\end{proof}
\addtocounter{proposition}{-1}
\begin{proposition}
Let $\Theta$ be a simplicial OPT satisfying $n$-local discriminability for some positive integer $n$. For every pair of systems $\rA,\rB\in\Sys{\Theta}$, every non-null extremal state $\rket{\lambda}_{\rA\rB}\in\Extst{AB}$ convexly refines the parallel composition of some pair of non-null extremal states $\rket{\rho}_{\sys{A}}\rket{\sigma}_{\sys{B}}\in\St{AB}$.
\end{proposition}
\begin{proof}

	By contradiction, let us suppose that there exist a pair of systems $\rA,\rB\in\Sys{\Theta}$ and an extremal  state $\rket{\lambda}_{\rA\rB}\in\Extst{AB}$ that does not convexly refine any  state of the form $\rket{\rho}_{\sys{A}}\rket{\sigma}_{\sys{B}}\in\St{AB}$ with $\rket\rho_\sys A$ and $\rket\sigma_\sys B$ non-null and extremal. Let us denote
	\begin{align}\label{eq:lambda}
	\rket{\lambda}_{\sys{AB}}\coloneqq \begin{tikzpicture}
	\begin{pgfonlayer}{nodelayer}
		\node [style=none] (0) at (0.25, -1) {};
		\node [style=none] (1) at (0.25, 1) {};
	\end{pgfonlayer}
	\begin{pgfonlayer}{edgelayer}
		\draw [style=dashed, bend right=90, looseness=1.25] (1.center) to (0.center);
	\end{pgfonlayer}
\end{tikzpicture}
\ .
	\end{align}
	The theory $\Theta$ satisfies $n$-local discriminability. The case $n=1$ is trivial, since all states are separable (see Theorem~\ref{thm:loc_discr_simplicial} in the main text) and $\rket{\lambda}_{\sys{AB}}$ must be vanishing. In the following we will then assume $n\geq 2$.
	Let us now define the following $(n+1)$-partite state:
	\begin{align}\label{eq:Psi_entangled}
	\rket{\Psi} \coloneqq \begin{tikzpicture}
	\begin{pgfonlayer}{nodelayer}
		\node [style=none] (0) at (0, 3.25) {};
		\node [style=none] (2) at (0, 2.5) {};
		\node [style=none] (3) at (0, 0.75) {};
		\node [style=none] (4) at (0, -0.75) {};
		\node [style=none] (5) at (0, -1.5) {};
		\node [style=none] (6) at (0, -3.25) {};
		\node [style=none] (8) at (0, 1.75) {\vdots};
		\node [style=none] (10) at (0, -2) {\vdots};
	\end{pgfonlayer}
	\begin{pgfonlayer}{edgelayer}
		\draw [style=dashed, bend right=60, looseness=1.25] (0.center) to (4.center);
		\draw [style=dashed, bend right=60, looseness=1.25] (2.center) to (5.center);
		\draw [style=dashed, bend right=60, looseness=1.25] (3.center) to (6.center);
	\end{pgfonlayer}
\end{tikzpicture}
\ \in \St{\left(A_1\cdots A_n\right)B_1\cdots B_n},
	\end{align}
	where the systems $\sys{A}_m$ and $\sys{B}_{m'}$ are copies of, respectively, $\sys{A}$ and $\sys{B}$ for every $m,m'\in\{1,\ldots,n\}$. By Lemma~\ref{lem:separable}, the state $\rket{\Psi}$ must be in the convex hull of some states refining the separable states $\rket\Lambda$ of $\sys{\left(A_1\cdots A_n\right)B_1\cdots B_n}$. The latter must be of one of the following two types 
	\begin{align}\label{eq:separable}
	\rket{\Lambda}=\begin{tikzpicture}
	\begin{pgfonlayer}{nodelayer}
		\node [style=none] (53) at (0.25, -3) {};
		\node [style=none] (54) at (0.25, 2.5) {};
		\node [style=none] (55) at (0.25, 1.25) {\vdots};
		\node [style=none] (64) at (0.25, 3.25) {};
		\node [style=none] (65) at (-0.5, 3.25) {};
		\node [style=none] (67) at (0.25, -0.75) {\vdots};
	\end{pgfonlayer}
	\begin{pgfonlayer}{edgelayer}
		\draw [bend left=270, looseness=0.50] (54.center) to (53.center);
		\draw [style=state] (65.center) to (64.center);
	\end{pgfonlayer}
\end{tikzpicture}
\quad,\ \begin{tikzpicture}
	\begin{pgfonlayer}{nodelayer}
		\node [style=none] (0) at (0.25, -1.5) {};
		\node [style=none] (1) at (0.25, 3.25) {};
		\node [style=none] (2) at (0.25, -0.75) {};
		\node [style=none] (3) at (0.25, 1.25) {};
		\node [style=none] (4) at (0.25, 0.5) {\vdots};
		\node [style=none] (5) at (0.25, -2.25) {\vdots};
		\node [style=none] (6) at (0.25, 2.5) {\vdots};
	\end{pgfonlayer}
	\begin{pgfonlayer}{edgelayer}
		\draw [bend left=270, looseness=0.75] (1.center) to (0.center);
		\draw [bend right=90, looseness=0.75] (3.center) to (2.center);
	\end{pgfonlayer}
\end{tikzpicture}
\quad ,
	\end{align}
	where in the first case $\sys{\left(A_1\cdots A_n\right)}$ is factorised from $\sys{B_1\cdots B_n}$, while in the second case there must exists a state of some proper subsystem of $\sys{B_1\cdots B_n}$ that is factorised.
	Then there exist coefficients $\alpha_i\in[0,1]$ such that:
	\begin{align}\label{eq:Psi_convex}
	\rket{\Psi}=\sum_{i}\alpha_i\rket{\phi_i},
	\end{align}
	where the $\rket{\phi_i}$ are non-null extremal states in the convex refinement of some separable state $\rket\Lambda$ of one of the two types in Eq.~\eqref{eq:separable}.
	By construction, in both cases 
	we can always find at least a subsystem $\sys{S}=\sys{A}_j\sys{B}_j$ of $\sys{A_1\cdots A_nB_1\cdots B_n}$ (now considered as a $2n$-partite system) such that the marginal state $\rbraket{e_\sys{\bar S}}{\Lambda}$  (where $\sys{\bar S}$ is the complementary subsystem of $\sys S$ in $\sys{ A_1\cdots A_n B_1\cdots B_n}$) is a product state of $\sys{A}_j\sys{B}_j$. On the other hand
	%
	%
	%
	the marginal state of $\rket{\Psi}$ on $\sys S$ is $\rbraket{e_\sys{\bar S}}{\Psi}=\rket{\lambda}_{\sys{A}_j\sys B_j}$. For each term on the r.h.s.~in Eq.~\eqref{eq:Psi_convex}, one can then find the above mentioned subsystem $\sys{S}$ and apply $\rbra{e_\sys{\bar S}}$ to both sides. This gives an equation of the form:
	\begin{align}\label{eq:alpha}
	\rket{\lambda}_{\sys{A}_j\sys B_j}= \alpha_i 
	\rket{\chi_i}_{\rA_j\rB_j}+ \rket{\omega_i}_{\sys{A}_j\sys B_j},
	\end{align}
	where $\rket{\chi_i}_{\rA_j\rB_j}$ is in the convex refinement of some product state by Lemma~\ref{lem:marginal}, and $\rket{\omega_i}_{\sys{A}_j\sys{B}_j}$ is a physical state. Since $\rket{\lambda}_{\sys{A}_j\sys{B}_j}$ is an extremal point of a simplex, its convex decomposition must be trivial, and then either $\alpha_i=0$ or
	\begin{align}
	\rket{\lambda}_{\sys{A}_j\sys B_j}=\rket{\chi_i}_{\rA_j\rB_j}.
	\label{eq:idass}
	\end{align}
	Finally, either $\alpha_i=0$ holds for every $i$, and then by direct inspection of the definitions~\eqref{eq:lambda} and~\eqref{eq:Psi_entangled} $\rket{\lambda}_\sys{AB}$ is vanishing, or 
	identity~\eqref{eq:idass} holds for some $i$, which is a contradiction.
\end{proof}

\end{document}